\documentclass[a4paper,UKenglish,cleveref]{lipics-v2021}

\usepackage{graphicx} % Required for inserting images\
\usepackage{floatrow}
\usepackage{xfrac}
\usepackage{svg}
\usepackage{algorithm}
\usepackage[noend]{algpseudocode}
\usepackage{caption}
\usepackage{subcaption}
\usepackage[T1]{fontenc}
\usepackage{babel}
\usepackage{tabularx, environ}
\usepackage{xspace}
\usepackage[most]{tcolorbox}
\usepackage[colorinlistoftodos, textwidth=2cm, textsize=tiny]{todonotes}

\newtheorem{branchingrule}{Branching Rule}

\newcommand{\N}{\mathbb{N}}
\newcommand{\problemfont}[1]{\textsc{#1}}
\newcommand{\scs}{\problemfont{Stable~Cutset}\xspace}
\newcommand{\set}[1]{\{#1\}}
\newcommand{\NPC}{$\mathcal{NP}$-complete\xspace}
\newcommand{\BigO}[1]{\mathcal{O}(#1)}
\newcommand{\OStar}[1]{\mathcal{O}^*(#1)}

\newcounter{case}

\newenvironment{case}[1][]{%
	\refstepcounter{case}%
	\medskip\noindent\textbf{Case~\thecase.} \emph{#1}\par\nopagebreak
}{\par}

\newcounter{subcase}[case]
\renewcommand{\thesubcase}{\thecase.\arabic{subcase}}

\newenvironment{subcase}[1][]{%
	\refstepcounter{subcase}%
	\medskip\noindent\textbf{Case~\thesubcase.} \emph{#1}\par\nopagebreak
}{\par}

\newcounter{subsubcase}[case]
\renewcommand{\thesubsubcase}{\thesubcase.\arabic{subsubcase}}

\newenvironment{subsubcase}[1][]{%
	\refstepcounter{subsubcase}%
	\medskip\noindent\textbf{Case~\thesubsubcase.} \emph{#1}\par\nopagebreak
}{\par}

\newtcolorbox{problembox}[2][]{colback=white!90!black,
	colframe=black, 
	coltitle=black,
	colbacktitle=white!80!black,enhanced,
	boxrule=0.8pt,
	attach boxed title to top left={yshift=-2mm, xshift=3mm},
	boxed title style={boxrule=0.8pt},
	title={#2},#1}

\nolinenumbers

\title{On Stable Cutsets in General and Minimum Degree Constrained Graphs}

\author{Mats Vroon}{Utrecht University, Department of Information and Computing Sciences}{m.vroon@students.uu.nl}{}{}
\author{Hans L. Bodlaender}{Utrecht University, The Netherlands}{H.L.Bodlaender@uu.nl}{https://orcid.org/0000-0002-9297-3330}{}

\authorrunning{Mats Vroon and Hans L. Bodlaender}

\Copyright{Mats Vroon and Hans L. Bodlaender}



\relatedversion{}

\hideLIPIcs

\date{ \today}
\begin{document}
	\maketitle
	\begin{abstract}
	A \emph{stable cutset} is a set of vertices $S$ of a connected graph, that is pairwise non-adjacent and when deleting $S$, the graph becomes disconnected. Determining the existence of a stable cutset in a graph is known to be \NPC. In this paper, we introduce a new exact algorithm for \scs. By branching on graph configurations and using the $\OStar{1.3645}$ algorithm for the \problemfont{(3,2)-Constraint Satisfaction Problem} presented by Beigel and Eppstein, we achieve an improved running time of $\OStar{1.2972^n}$. 	
    %MV: word change, om herhaling van 'present' te voorkomen.
    
	In addition, we investigate the \scs problem for graphs with a bound on the minimum degree $\delta$. First, we show that if the minimum degree of a graph $G$ is at least $\frac{2}{3}(n-1)$, then $G$ does not contain a stable cutset. Furthermore, we provide a polynomial-time algorithm for graphs where $\delta \geq \tfrac{1}{2}n$, and a similar kernelisation algorithm for graphs where $\delta = \tfrac{1}{2}n - k$.
	Finally, we prove that \scs remains \NPC for graphs with minimum degree $c$, where $c > 1$. We design an exact algorithm for this problem that runs in $\OStar{\lambda^n}$ time, where $\lambda$ is the positive root of $x^{\delta + 2} - x^{\delta + 1} + 6$. This algorithm can also be applied to the \problemfont{3-Colouring} problem with the same minimum degree constraint, leading to an improved exact algorithm as well.
\end{abstract}
	\section{Introduction}
A \textit{stable cutset} in a connected graph $G = (V, E)$ is a subset of vertices $S \subseteq V$ that is pairwise non-adjacent ($S$ is a \textit{stable set}) and when deleting these vertices from $G$, the graph becomes disconnected ($S$ is a cutset). This gives the following decision problem:

\begin{problembox}{\scs}
	\begin{tabbing}
		\hspace{1.5cm} \= \kill  % Set tab stops
		\textit{Instance:} \> A connected graph $G$. \\
		\textit{Question:} \> Does $G$ have a stable cutset?
	\end{tabbing}
\end{problembox}

% Perfect graph shit
Applications of \scs were already shown in connection with perfect graphs \cite{StableSetBondingInPerfectGraphsAndParityGraphs, ColoringGraphsWithStableCutsets}. Tucker~\cite{ColoringGraphsWithStableCutsets} proved that $G$ is $k$-colourable if and only if all
$G_i \cup S$ are $k$-colourable, for all components of $G\setminus S$ for a stable cutset $S$ of $G$ with no vertex of $S$ lying on a hole.

% NP-completeness
\scs has been proven \NPC first by Chv\'{a}tal \cite{RecognizingDecomposableGraphs}.
The \scs problem is related to the problem to find a \textit{matching cut}. A \textit{matching cut} in a graph $G=(V,E)$ is a set $M\subseteq E$ of edges that form a matching such that $G$ becomes disconnected by deleting $M$ from it.
It has been shown that determining whether a graph with maximum degree 4 admits a \textit{matching cut} %(\textit{decomposable}) 
is \NPC. 
Clearly, a graph with minimum degree 2 has a matching cut if and only if the \textit{line graph} $L(G)$ has a stable cutset. A line graph $L(G)$ is constructed by creating a vertex for every edge in $G$, and these vertices are adjacent if and only if the corresponding edges in $G$ share a common vertex \cite{LineGraphsAndLineDigraphs}. Now, consider that we have a stable cutset $S$ in $L(G)$. The edges of $G$ that correspond to the vertices in $S$, denoted as $S'$, will not have a vertex in common, and will separate the graph $G$, hence $S'$ is a matching cut of $G$. 
Further investigation was done by
Brandst\"{a}dt et al. \cite{OnStableCutsetsInGraphs}, who showed that \scs remains \NPC even when restricted to $K_4$-free graphs. Also, \scs is \NPC in line or planar graphs with maximum degree 5 \cite{OnStableCutsetsInLineGraphs, OnStableCutsetsInClawFreeGraphsAndPlanarGraphs}. 

% Polynomial classes
Despite the $\mathcal{NP}$-completeness of \scs, there are known graph classes where `fast' solutions are possible. Trivially, $K_3$-free graphs and graphs with connectivity number 1, have a stable cutset. In a $K_3$-free graph every vertex has stable neighbours which form a stable cutset,
and a cut vertex in a graph of connectivity 1 is a stable cutset on its own.
Chen and Yu~\cite{ANoteOnFragileGraphs} proved that graphs with at most $2n - 4$ edges contain a stable cutset. This bound is 
tight as there exists a graph without a stable cutset with $2n - 3$ edges~\cite{ExtremalGraphsHavingNoStableCutsets}. But, \scs can be solved for graphs with $2n - 3$ edges in polynomial-time~\cite{ExtremalGraphsHavingNoStableCutsets}. Furthermore, polynomial-time algorithms are constructed for \textit{HHD}-free graphs, brittle graphs and hole-free graphs~\cite{OnStableCutsetsInGraphs}. Also, polynomial-time reduction algorithms are constructed for graphs with maximum degree 3, and line graphs with maximum degree 4~\cite{OnStableCutsetsInLineGraphs}. Both algorithms reduce the input graph to such a small graph that solving becomes trivial, or it is reduced to a different graph class where a known algorithm applies. Le~et~al.~\cite{OnStableCutsetsInClawFreeGraphsAndPlanarGraphs} showed that there exists a polynomial-time algorithm for claw-free graphs with maximum degree 4. This combined with a reduction rule created a polynomial-time algorithm for (claw, $K_4$)-free graphs. They also show that there exists a polynomial-time algorithm on claw-free planar graphs. 

% FPT algoritmes
A small amount of work have been done on the parameterized complexity of \scs. 
In 2013, Marx et al.~\cite{FindingSmallSeparatorsInLinearTimeViaTreewidthReduction} showed that
\scs, parameterized by solution size, is \textit{fixed-parameter tractable}. Recently, Rauch et al. have proven that \scs is FPT when parameterized by dual of maximum degree, size of a dominating set, distance to $P_5$-free graphs, and more \cite{ExactParameterizedAlgorithmsForIndCutsetProblem}. Furthermore, Kratsch and Le \cite{OnPolyKernelizationStableCutset} have shown that there exists \textit{polynomial kernelisations} when parameterized by modulators to a single clique, to a cluster graph, and when parameterized on a twin cover. Also, they show for a few parameters that \textit{polynomial kernelisation} is not possible.

%Exact algoritmes
Exact algorithms for this problem have not yet been extensively studied. Only Rauch et al. \cite{ExactParameterizedAlgorithmsForIndCutsetProblem} have presented an exact algorithm for \scs.
They prove that every stable set $S'$, which is a superset of a stable cutset $S$ in $G$, is also a stable cutset in $G$. With that knowledge, enumerating all the maximal stable sets and checking for separation would suffice to check whether the graph has a stable cutset. Additionally, they show how this maximal stable cutset can be reduced to a minimal stable cutset.
This gives an algorithm with a running time of $\OStar{3^{n/3}}$.

Kratsch and Le \cite{AlgorithmsSolvingTheMatchingCutProblem} gave the first exact branching algorithm for the closely related \problemfont{Matching Cutset} problem, with a time complexity of $\OStar{2^{n/2}} = \OStar{1.4143^n}$. This was improved in 2020 by Komusiewicz et al. \cite{MatchingCutKernelizationSingleExponentialTimeFPTAndExactExponentialAlgorithm}. By reducing \problemfont{Matching Cut} to \problemfont{3-SAT},
they showed that \problemfont{Matching Cutset} can be solved in $\OStar{1.3280^n}$ time. Using a
branching algorithm and not relying on \problemfont{SAT}, an algorithm with $\OStar{1.3803^n}$ time was obtained.

Another problem that is similar to \scs is \problemfont{3-Colouring}: in both of these problems, we assign to each vertex one colour out of a set of size three. The constraints do differ slightly. 
In \problemfont{3-Colouring}, adjacent vertices cannot have the same colour. However, in \scs, we have colour classes $A$, $B$ and $S$, were vertices in $S$ cannot be adjacent, vertices
in $A$ cannot be adjacent to vertices in $B$,
but there can be edges between vertices that are both in $A$ or both in $B$.
The current fastest exact algorithm for \problemfont{3-Colouring} is found by Meijer \cite{3ColoringInTime1.3217^n} with time complexity of $\OStar{1.3217^n}$. This algorithm is based on the algorithm presented by Beigel and Eppstein \cite{3ColoringInTime1.3289^n} for the \problemfont{3-Colouring} problem. They presented an $\OStar{1.3645^n}$ algorithm to solve the \problemfont{(3,2)-Constraint Satisfaction Problem}, to which \problemfont{3-Colouring} can be modelled. 
%By using some graph reduction techniques, they gave an even further advanced algorithm for \problemfont{3-Colouring} with a time of $\OStar{1.3289^n}$.
By enumerating all possible colourings for a small set of vertices and using a (3,2)-CSP reduction rule, an exact algorithm with a time complexity of $\OStar{1.3289^n}$ is achieved.
%MV: Ik bedenk me nu dat de 3-kleurings algoritme niet perse graph reduction technieken gebruikt. Verandert naar accurater beschrijving wat in dat paper gebeurd.

Unlike the research for \problemfont{Matching Cutset} and \problemfont{3-Colouring}, the literature has not yet explored \scs problems when restricted to graphs with a certain minimum degree. 
For example, Chen et al.~\cite{MatchingCutInGraphsWithLargeMinimumDegree} presented an algorithm that solves the \problemfont{matching cutset} problem for graphs with minimum degree $\delta \geq 3$ in $\OStar{\lambda^n}$, where $\lambda$ is the positive root of the polynomial $x^{\delta+1} - x^\delta - 1$. Furthermore, for \problemfont{3-Colouring}, an algorithm based on dominating sets is presented that has a time complexity of $\mathcal{O}(1.2957^n)$ for $\delta \geq 15$ and $\mathcal{O}(1.1^n)$ for $\delta \geq 50$ \cite{3ColoringWithMinDegC}.

This paper investigates \scs in both general graphs and graphs with minimum degree constraints. We begin in Section~\ref{sec:preliminaries} by introducing the necessary definitions and notations. An annotated graph model for \scs is developed in Section~\ref{sec:modelingSCS}, where, in addition, we demonstrate how this model can be reduced to a \problemfont{(3,2)-CSP} instance. This model is used in all subsequent sections of the paper.
The first algorithmic contribution is an improved exact algorithm for general graphs which is discussed in Section~\ref{sec:generalExactAlg}. Then we shift to graphs with minimum degree constraints. Section~\ref{sec:minDegCN} presents different results for \scs in graphs with minimum degree $\delta = c \cdot n$, where $c < 1$. In Section~\ref{sec:minDegC} we show a $\mathcal{NP}$-completeness proof and an exact algorithm for \scs in graphs with minimum degree $c$, where $c > 1$. This exact algorithm is with minor adaptations also applied to \problemfont{3-Colouring} in graphs with the same minimum degree condition in Section~\ref{sec:3ColMinDegC}. We conclude in Section~\ref{sec:conclusion} with a summary of our findings.

%Ik wil even de structuur hier benoemen:
%- Beschrijving van het probleem
%- Toepassing op perfect graphs en coloring decomposition
%- NP-completeness
%   - Eerste bewijs door matching cutset
%- Bekende polynomiale algoritmes voor bepaalde classes
%- Bekende FPT algoritmes 
%- Bekent Exact Algoritme
%   - Stable cutset
%   - Matching cutset
	\section{Preliminaries}\label{sec:preliminaries}
Unless stated otherwise, we consider undirected simple graphs $G = (V,E)$. The number of vertices is denoted as  $n = |V|$ and the number of edges as $m = |E|$. The \textit{(open) neighbourhood} of a vertex $v$, is the set of vertices which are the neighbours of $v$. This is denoted by $N(v)$. The \textit{closed neighbourhood} is defined as $N[v] = N(v) \cup \set{v}$. Let $U$ be a subset of the vertices $V$. We define the neighbourhood of a set as $N(U) = \bigcup_{v \in U} N(v) \setminus U$. The graph induced by the vertices of $U$ is denoted as $G[U]$. We write $G-U$ for the graph $G[V\setminus U]$. The vertex set $U$ is a \emph{clique} if every pair of vertices in $U$ is adjacent. A \emph{triangle} is a clique of size three.  A \textit{cutset} is defined as a subset of vertices $S \subseteq V$ such that the graph $G-S$ is disconnected. If a subset of vertices are pairwise non-adjacent then the set is called \textit{stable}. A \textit{stable cutset} is a subset of vertices that is \textit{stable} and a \textit{cutset}.

% Insert problem description for (3,2)-CSP
The \problemfont{Constraint Satisfaction Problem} (CSP) is the problem whether it is possible to assign values to variables such that all the constraints are satisfied. Each variable has a \textit{domain} of values that can be assigned to the variable. In the \problemfont{(3,2)-CSP} problem, every variable has a domain of size at most three, and each constraint involves at most two variables. Thus, the constraints specify which combinations of values are not allowed for certain pairs of variables.

A problem is \textit{fixed-parameter tractable} (FPT) when there is an algorithm that decides
the problem in 
$f(k) \cdot |(x,k)|^c$ time, for each instance $x$ and parameter $k$, where $c$ is a constant independent of $x$ and $k$, and $f$ is a computable function~\cite{ParameterizedAlgorithms}.
A \textit{kernelisation algorithm} is an algorithm that given an instance $(x, k)$, works in polynomial time and returns an equivalent instance $(x', k')$, with
$|x'|$ and $k'$ bounded by computable functions of $k$~\cite{ParameterizedAlgorithms}. It is a \textit{polynomial kernelisation} if the \textit{size} of the reduced instance $|x'|$ has a polynomial upper bound in $k$. 
The notation $\mathcal{O}^*$, is a modified big-O notation where polynomial bounded factors are left out. More formal, $f(n) = \BigO{g(n)\cdot poly(n)} = \OStar{g(n)}$ \cite{ExactExponentialAlgorithmsBook}. 
	\section{Modelling the \scs problem}\label{sec:modelingSCS}
In this section, we describe the manner
we model the stable cutset problem as an annotation problem, as we will use throughout this paper. We then present several properties and rules that will be applied significantly in the subsequent sections. After that, we show how to transform a stable cutset instance to a (3,2)-CSP instance, leading to our first result. Finally, we discuss a lemma proven by Beigel and Eppstein \cite{3ColoringInTime1.3289^n} and how this is helpful in reducing the instance size.

\subsection{\scs as an annotation problem}\label{sec:scsAsAnnot}
% Structuur hoe ik het voor ogen zie:
% a) Formal definition van literatuur (miss 1 of ander boek)?
% b) Literatuur dingen voor 3-coloring en match. cutset 
% c) In wat voor format wij labeling gaan gebruiken (possible options set)
% d) Regels en properties voor deze definition

Annotated graphs are types of graphs that add information to their vertices or edges. For example, in \problemfont{Colouring} problems, an annotation function $c: V \rightarrow C$ assigns a color from a set $C$ to each vertex \cite{GraphTheoryBook}. Similarly, the \problemfont{Matching Cutset} problem can be modelled as an annotated graph, where vertices are labelled to indicate their belonging to one of the two components \cite{MatchingCutKernelizationSingleExponentialTimeFPTAndExactExponentialAlgorithm, MatchingCutInGraphsWithLargeMinimumDegree}. 

The \scs problem can be represented as an annotated graph, where each vertex is assigned with one of the three possible labels $L = \set{A, B, S}$. The vertices of the two connected components that need to be separated are labelled with $A$ and $B$, and the vertices of the stable cutset are labelled with $S$. An annotated graph $G$ contains a stable cutset if it satisfies the following adjacency and non-empty constraints:
\begin{itemize}
	\item No pair of adjacent vertices is labelled with $A$ and $B$, or both with $S$.
	\item Each label must appear on at least one vertex. 
\end{itemize}

In this paper, we use an annotated graph that keeps even more information. Per vertex we keep track of the  set of labels that `still are possible', $p: V \rightarrow \mathcal{P}(L)$. Before running some algorithm, all the vertices will be annotated with the full set $L$. An algorithm on this annotated graph is done if all the vertices have one `possible' label left, and the resulting labelling satisfies the above described adjacency and non-empty constraints. An example representation of the annotated graph can be seen in Figure \ref{fig:annotatedgraphexample}. Using this representation, we can already find some general properties for solving \scs.  

\renewcommand{\labelenumi}{(\roman{enumi})}
\begin{lemma}\label{lem:annGraphRules}
	Consider the graph $G = (V,E)$ with annotation function $p: V \rightarrow \mathcal{P}(L)$. Given some edge $(u, v) \in E$: 
	\begin{enumerate}
		\item If $p(u) = \set{A}$ then $B \notin p(v)$. 
		\item If $p(u) = \set{B}$ then $A \notin p(v)$. 
		\item If $p(u) = \set{S}$ then $S \notin p(v)$. 
	\end{enumerate}
\end{lemma}
%MV: Is deze styling van romeinse cijfers enumeration wel de juiste?

\begin{proof}
	The properties above are based on the definition of the stable cutset problem. Since the goal is to separate at least two connected components, the vertices labelled with $\set{A}$ can not be adjacent to vertices labelled with $\set{B}$, and vice versa. Additionally, to ensure that the cutset is stable, vertices labelled with $\set{S}$ can not be adjacent to other vertices labelled with $\set{S}$.
\end{proof}

By using some algorithm and the above rules, we can get to a state where there exists some vertex $v$ in $G$ such that $p(v) = \emptyset$. This could occur for example when $v$ has three different neighbours $x, y, z$ that are labelled with $\set{A}$, $\set{B}$ and $\set{S}$ respectively. While the unannotated graph may still admit a stable cutset, we know that the current annotated graph does not contain one. This formalised in a lemma gives:

\begin{lemma}
	Consider the graph $G = (V,E)$ with annotation function $p: V \rightarrow \mathcal{P}(L)$. If there exists some vertex $v \in V$ with $p(v) = \emptyset$, then there does not exist a stable cutset in the annotated graph.
\end{lemma}

\begin{figure}[h]
	\centering
	\includegraphics[width=0.5\linewidth]{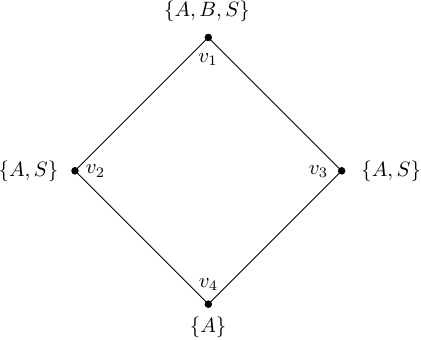}
	\caption{Example of some annotated graph $G$. Note, that $p(v_4) = \set{A}$ and therefore $B \notin p(v_2)$ and $B \notin p(v_3)$ by Lemma \ref{lem:annGraphRules}.}
	\label{fig:annotatedgraphexample}
\end{figure}

\subsection{\scs as \problemfont{(3,2)-CSP}}\label{sec:annTo32CSP}
In the context of research on the \problemfont{3-Colouring} problem, Beigel and Eppstein \cite{3ColoringInTime1.3289^n} showed that \problemfont{(3,2)-CSP} can be solved in $\OStar{1.3645^n}$, where $n$ is the amount of variables. Below, we will show how our instance for stable cutset can be transformed into a (3,2)-CSP instance. Because the amount of vertices will be the same as the amount of variables in the corresponding (3,2)-CSP, it follows that \scs can be solved in $\OStar{1.3645^n}$

Let us consider the annotated graph $G = (V,E, p)$ as described in Section~\ref{sec:scsAsAnnot}. To construct the (3,2)\nobreakdash-CSP instance, we create a variable for every vertex $v$, where the domain for that variable is $p(v)$. This can be done because the semantics of the annotation of a vertex in $G$ are the same as the domain of a variable in (3,2)-CSP. If we look at the rules of Lemma~\ref{lem:annGraphRules}, we can see that they correspond to a constraint with two variables. So, for every edge $(u,v) \in E$, we add $\set{((u, A), (v, B)), ((u, B), (v, A)), ((u, S), (v, S))}$ to the constraint set of the (3,2)-CSP instance. An example of this instance transformation can be seen in Figure \ref{fig:32cspgraphexample}. 

Before we send our instance to the (3,2)-CSP solver of Beigel and Eppstein, we need to ensure that the non-emptiness constraints described in Section~\ref{sec:scsAsAnnot} are satisfied. Otherwise, we could simply label all the vertices with $A$, and no constraints of the (3,2)-CSP instance are violated. 
To prevent this, we pick a random vertex $u$ in the graph and label it with $p(u) = \set{A}$. We then run the algorithm on all vertices $v \in V \setminus \set{u}$ and prematurely label $p(v) = \set{B}$. If the algorithm succeeds in one of the $O(n)$ instances, we are done and know that the graph has a stable cutset. If this is not the case, then $u$ may be part of the stable cutset. In that case, we label $u$ with $p(u) = \set{S}$ and pick a vertex $x \in N(u)$. We label this vertex with $p(x) = \set{A}$. Now we repeat the process by running the algorithm for all vertices $v \in V \setminus \set{u, x}$ and prematurely labelling $p(v) = \
\set{B}$. If the algorithm fails on all these new $O(n)$ instances, then the graph does not contain a stable cutset. Otherwise, the graph does contain one. This gives our first algorithm for \scs; in later sections, we show how to improve on this bound by using more advanced techniques.

% To prevent this, we have to run the algorithm for every pair of vertices $(u, v) \in V \times V$ and prematurely label $p(u) = \set{A}$ and $p(v) = \set{B}$ before transforming it to a (3, 2)-CSP instance. This requires running the (3,2)-CSP solver on $O(n^2)$ different annotated graph instances. The transformation from the annotated graph instance to (3,2)-CSP instance, can be done in polynomial time.

\begin{theorem}
	There is a $\OStar{1.3645^n}$ algorithm
    for \scs.
\end{theorem}

\begin{figure}
	\centering
	\includegraphics[width=0.5\linewidth]{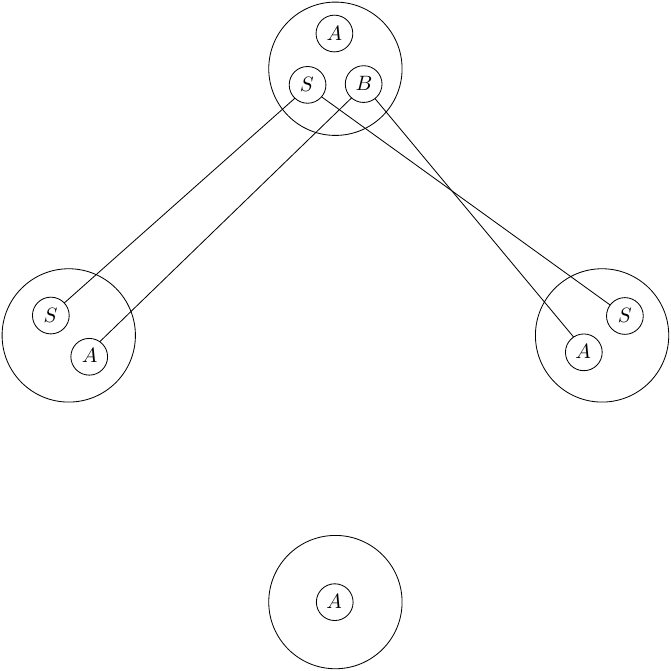}
	\caption{The annotated graph from Figure \ref{fig:annotatedgraphexample} transformed to a \problemfont{(3,2)-CSP} instance. The lines in the representation indicate the constraints.}
	\label{fig:32cspgraphexample}
\end{figure}

\subsection{Reducing instance size}
Besides the (3,2)-CSP solver of Beigel and Eppstein, we are also using the following lemma of the same paper:

\begin{lemma}[\cite{3ColoringInTime1.3289^n}]\label{lem:decreaseValuesBE}
	Let $v$ be a variable in a \problemfont{($a$, 2)-CSP} instance, such that only two values are allowed at $v$. Then we can find an equivalent \problemfont{($a$, 2)-CSP} instance with one fewer variable.
\end{lemma}

This is useful due to the correspondence between the semantics of the annotation function $p$ for some vertex, and the amount of allowed values for the corresponding variable in the (3,2)-CSP instance. Therefore, if we have a vertex $v$ with $|p(v)| \leq 2$, then the corresponding variable in the (3,2)-CSP instance can be deleted, thereby reducing the instance size. 

	\section{Exact algorithm}\label{sec:generalExactAlg}
In this section, we describe how we are reducing the time bound for the general case even further. The general idea of the algorithm is to find vertex sets in the graph that allow branching in such a way that for many vertices in the set, the number of possible labels is reduced to less than three. For example, when we have a clique induced in the graph $G$, then the vertices of the clique can only be labelled with either $\set{A, S}$ or $\set{B, S}$. This is because a vertex labelled with $A$ can not be adjacent to a vertex labelled with $B$. If we branch on these two options, the possible labels for vertices of the clique are reduced to less than three. Eventually, we have a branching tree with in the leaves, annotated graph instances that can be greatly reduced by Lemma~\ref{lem:decreaseValuesBE} when transformed to (3,2)-CSP instances.

First, we will present a intuitive `stupid' algorithm that uses the above described idea. We then refine the algorithm to make smarter choices in finding the vertex sets, which will reduce the time complexity even further.

\subsection{Initial approach}\label{sec:exactAlgInitialApproach}
The input of the algorithm is the annotated graph $G = (V, E, p)$ as described in Section \ref{sec:scsAsAnnot}. We describe the steps of the algorithm and establish its correctness through a sequence of claims followed by the branching rules. The algorithm starts by checking whether the neighbourhood of each vertex $v$ in $G$ forms a stable set. If that is the case then we are done, because $N(v)$ separates $v$ from $V \setminus N[v]$. After this step we claim,

\begin{claim}\label{clm:vertexPartOfTraingle}
	Every vertex in $G$ is part of some triangle.
\end{claim}

We greedily construct the family of vertex sets $\mathcal{H}$ by adding as many pairwise disjoint triangles from $G$ as possible, such that $\mathcal{H}$ is maximal. Let $F$ denote the set of vertices that is not part of some triangle in $\mathcal{H}$. 

\begin{claim}\label{clm:fAdjacentToTriangle}
	Given the above definitions, all the vertices in $F$ are adjacent to a vertex of some triangle in $\mathcal{H}$.
\end{claim}

\begin{proof}
	Consider the vertex $f \in F$. By Claim~\ref{clm:vertexPartOfTraingle}, we know that $f$ is part of some triangle $T$ in $G$. By the definition of $F$, we know that $T \notin \mathcal{H}$. Since $\mathcal{H}$ is maximal in $G$, at least one of the two other vertices in $T$ must belong to a triangle $R \in \mathcal{H}$, because otherwise $T$ would have been in $\mathcal{H}$. Therefore, $f$ is adjacent to a vertex of some triangle in $\mathcal{H}$.
\end{proof}

To ensure that every vertex of $G$ is included in some vertex set of $\mathcal{H}$, we add each vertex $f \in F$ to a vertex set (triangle) that contains a neighbour of $f$. The existence of such a vertex set is guaranteed by Claim~\ref{clm:fAdjacentToTriangle}. As a result, each vertex set $W \in \mathcal{H}$ now consists of a triangle $T$ together with a subset of its neighbours, meaning $W = T \cup (W \setminus T)$ where $W \setminus T \subseteq N(T)$.

Our algorithm uses two branching strategies, which both reduce the amount of possible labels for specific vertices. Because of this reduction, we can delete the corresponding variables when the annotated graph instances are transformed into a \problemfont{(3,2)-CSP} instance by Lemma~\ref{lem:decreaseValuesBE}. 

\renewcommand{\labelenumi}{\arabic{enumi}.}
\begin{enumerate}
	\item \textit{Branching on triangles:} As mentioned above, we can branch on cliques which gives two annotated graph instances. Therefore, the possible labels of the vertices in the triangle will be reduced.
	\item \textit{Branching on neighbourhoods:} Alternatively, a configuration on which we can branch is $N[v]$, for some vertex $v \in V$. By branching on the possible labels of $v$, we obtain three instances. In these instances $|p(v)| = 1$, and all the vertices $u \in N(v)$ have $|p(u)| = 2$ because of Lemma~\ref{lem:annGraphRules}. 
\end{enumerate}
\renewcommand{\labelenumi}{(\roman{enumi})}

We apply specific branching rules to each vertex set $W \in \mathcal{H}$. The choice of the branching rule depends on the vertex $v \in T$ with the most neighbours in $W \setminus T$. Only a subset of vertices of $W$ are used in the branching rules.  The remaining vertices become standard variables for the \problemfont{(3,2)-CSP} solver.

\begin{branchingrule}\label{br:closedNeighborhoodrule}
	Let $v$ be the vertex of $T$ where $|N(v) \cap (W \setminus T)|$ is the largest. If $|N(v) \cap (W \setminus T)| \geq 2$, then
	\begin{enumerate}
		\item Assign $p(v) = \set{A}$ and apply Lemma~\ref{lem:annGraphRules} to all vertices in $N(v) \cap W$.
		\item Assign $p(v) = \set{B}$ and apply Lemma~\ref{lem:annGraphRules} to all vertices in $N(v) \cap W$.
		\item Assign $p(v) = \set{S}$ and apply Lemma~\ref{lem:annGraphRules} to all vertices in $N(v) \cap W$.
	\end{enumerate}
	The remaining vertices are handled by the (3,2)-CSP solver.
\end{branchingrule}

\begin{branchingrule}\label{br:triangleRule}
	If Branching Rule \ref{br:closedNeighborhoodrule} is not applicable, then
	\begin{enumerate}
			\item Assign $p(v) = \set{A, S}$ for all vertices $v \in T$.
			\item Assign $p(v) = \set{B, S}$ for all vertices $v \in T$.
	\end{enumerate}
	The remaining vertices are handled by the (3,2)-CSP solver.
\end{branchingrule}

\begin{figure}
	\centering
	\includegraphics[width=0.8\linewidth]{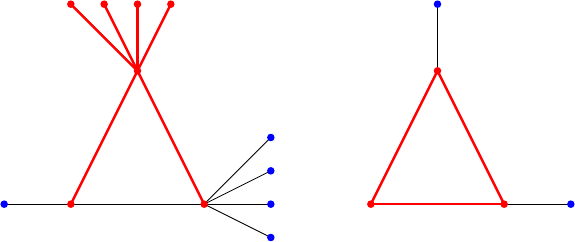}
	\caption{Branching Rule 1 and 2, respectively. The vertices highlighted in red correspond to the configurations being branched on, while the vertices in blue represent those that remain as standard variables in the (3,2)-CSP instance.}
	\label{fig:branchingrules}
\end{figure}

Once the algorithm has branched on all the vertex sets of $\mathcal{H}$, the leaves of the branching tree contain annotated graphs where many vertices are annotated with at most two possible labels. The corresponding variables of these vertices will be deleted when the annotated graph is converted to a (3,2)-CSP instance.
We calculate the total running time of the algorithm by finding the maximum \textit{cost} of any vertex in the graph. The cost of a vertex is defined as the product $p$ of factors involving vertices of a disjoint vertex set, spread evenly over the vertices. So, for some disjoint vertex set $W \in \mathcal{H}$, the cost of any vertex in this set is $c = p^{1/|W|}$. The vertex set for which $c$ is the largest will define the total time: $\OStar{c^n}$. The factors in this case are 2 or 3 if a triangle or closed neighbourhood is used, and a factor of $1.3645$ is used per vertex that will be converted to a normal variable in the (3,2)-CSP instance.

\begin{lemma}
	The worst-case vertex cost among all vertex sets is 1.3112.
\end{lemma}

\begin{proof}
	We begin by analysing the vertex costs associated with Branching Rule \ref{br:closedNeighborhoodrule}. Let $v$ be the vertex of $T$ where $|N(v) \cap (W \setminus T)|$ is the largest. We define the numerical variable $x = |N(v) \cap (W \setminus T)|$. The other vertices in $T$ have a smaller than or equal amount of neighbours in $W \setminus T$. Therefore, the worst-case vertex cost for this rule is defined as $(3 \cdot 1.3645^{2x})^{\tfrac{1}{3+3x}}$. We want this value to be lower than $1.3$. By solving the inequality $(3 \cdot 1.3645^{2x})^{\tfrac{1}{3+3x}} < 1.3$, we find that for a positive $x$, the inequality holds whenever $x > 1.8821$. Because Branching Rule \ref{br:closedNeighborhoodrule} is conditioned on $x \geq 2$, it follows that the vertex cost for any vertex set to which this rule applies will always be less than $1.3$.
	
	Following this, we examine Branching Rule \ref{br:triangleRule}. Among the four distinct vertex sets to which this rule applies, the worst-case vertex cost is retrieved when each vertex in $T$ is adjacent to one unique vertex in $W \setminus T$. This results in a vertex cost of $(2 \cdot 1.3645^3)^{1/6} \approx 1.3112$.  
\end{proof}

%The cost of a vertex in the left example of Figure \ref{fig:branchingrules} will therefore be: $(2 \cdot 1.3645^2)^{1/5} \approx 1.3008$. The other examples have a vertex cost of $1.2786$ and $1.2325$ respectively. The vertex set $W$ with the worst vertex cost is when every vertex of the triangle $T \subseteq W$ is adjacent to exactly one vertex of $W \setminus T$. In this case, the vertex cost is $(2 \cdot 1.3645^3)^{1/6} \approx 1.3112$. 

\begin{theorem}
	There is a $\OStar{1.3112^n}$ algorithm for \scs.
\end{theorem}

\subsection{Smarter vertex sets}
In the above algorithm, we used random assignments of the vertices in $F$ to get all the vertices of $G$ involved in the vertex set family $\mathcal{H}$. With this, the slowest vertex set can be achieved which defines the total time of the algorithm. Below, we describe a polynomial time algorithm that will divide the vertices of $G$ into disjoint vertex sets. This algorithm ensures that the slowest and second slowest vertex sets are avoided. The slowest vertex set, as described in Section~\ref{sec:exactAlgInitialApproach}, is the vertex set $W$ where all the vertices of the triangle $T \subseteq W$ are adjacent to one unique vertex of $W \setminus T$. The second slowest is the set $W$ where two vertices of the triangle $T \subseteq W$ are adjacent to one unique vertex of $W \setminus T$ and the third vertex of $T$ is not adjacent to any vertex of $W \setminus T$. In the rest of this section we call these vertex sets the \textit{slow vertex sets}, and all others the \textit{fast vertex sets}. 

Consider the annotated graph $G = (V, E, p)$ as described in Section~\ref{sec:scsAsAnnot} as input. Again, we first check whether any vertex $v$ has stable neighbours. If that is the case, then the algorithm terminates as mentioned in Section~\ref{sec:exactAlgInitialApproach}. After this step, we know that Claim~\ref{clm:vertexPartOfTraingle} holds in $G$. 

The algorithm builds the family of vertex sets $\mathcal{H}$ by repeatedly identifying the vertex sets, adding them to $\mathcal{H}$ and removing them from the graph $G$. To guarantee that this process can progress until the graph is empty, an invariant is used: \emph{Every vertex in $G$ is part of some triangle ($K_3$)}. We use the function $r_G(U)$ that denotes the set of vertices in $G - U$ that are no longer part of any triangle. When it is clear from the context which graph is being used, we will write $r(U)$ without explicitly naming the graph.

The procedure begins by selecting an arbitrary triangle $T$ in the graph, which must exists due to the invariant. Based on the number of vertices in $r(T)$ and their adjacency to the vertices of $T$, the procedure proceeds via a case analysis. Each case specifies how to identify the vertex sets that will be added to $\mathcal{H}$, and deleted from $G$. After presenting the case analysis, we prove that only fast vertex sets will be added to $\mathcal{H}$, and that the invariant is preserved. 

\begin{case}[Two or more vertices in $r(T)$ are adjacent to the same vertex in $T$, or $|r(T)| = 1$.]\label{case:exactAlgEasyCase}
	This is the easy case. We add $W = T \cup r(T)$ to $\mathcal{H}$ and update the graph to $G - W$.
\end{case}

\begin{case}[$|r(T)| = 3$ and every vertex in $r(T)$ is adjacent to a different vertex in $T$.]\label{case:exactAlgThreePerCorner}
	This case is visualized in Figure~\ref{fig:exactalgcase2}. By the invariant we know that all the vertices of $r(T)$ are part of some triangle. Let $T_1$, $T_2$ and $T_3$ be those triangles, and we denote the vertices of triangle $T_i$ ($i \in \set{1,2,3}$) by $v_{i,1}, v_{i,2}, v_{i,3}$. By the definition of $r(T)$, each triangle $T_i$ must contain one vertex from $T$. Let $v_{i, 1}$ be that vertex, and let $v_{i, 2}$ be the vertex in $r(T)$. There are two subcases on $v_{i, 3}$ of the $T_i$ triangles:
	
	\begin{subcase}[The vertices $v_{1,3}$, $v_{2,3}$ and $v_{3,3}$ are distinct.]\label{case:thrdVtxsDistinct}
		This case is visualized on the left of Figure~\ref{fig:exactalgcase2}. We create the following vertex sets: $W_1, W_2, W_3$, where $W_i = T_i - r(T_i)$. These three vertex sets are added to $\mathcal{H}$ and we update the graph to $G - \bigcup_i^3 W_i$.
	\end{subcase}
	
	\begin{subcase}[Two or more of the vertices $v_{1,3}$, $v_{2,3}$ and $v_{3,3}$ are the same.]\label{case:commonThirdVtx}
		This case is visualized on the right of Figure~\ref{fig:exactalgcase2}. We call the common third vertex $u$. Since $u$ is not part of $r_G(T)$, it must belong to some triangle $P$ that does not include any vertices from $T$. So, we create $W_1 = P \cup r_G(P)$ and update the graph to $G' = G - W_1$. Now, we make a second vertex set $W_2 = T \cup r_{G'}(T)$ and update the graph again to $G' - W_2$. These two vertex sets are added to $\mathcal{H}$.
	\end{subcase}
	
	\begin{figure}
		\centering
		\includegraphics[width=0.8\linewidth]{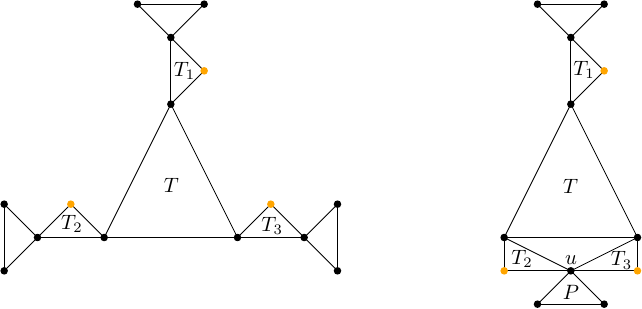}
		\caption{Visualization of Case \ref{case:exactAlgThreePerCorner}. The vertices in orange correspond to the vertices in $r(T)$.}
		\label{fig:exactalgcase2}
	\end{figure}
\end{case}

\begin{case}[$|r(T)| = 2$ and every vertex in $r(T)$ is adjacent to a different vertex in $T$.]\label{case:exactAlgTwoVtxToDiffCorn}
	Let $\set{a,b}$ denote the vertices of $r(T)$, and let $\set{v_1, v_2, v_3}$ denote the vertices of $T$. There are two cases to consider within this case.
	
	\begin{subcase}[$a$ is adjacent to $v_1$ and $v_2$, and $b$ is only adjacent to $v_3$.]\label{case:vtxConnectToTwoOfT}
		Let $A$ be the triangle $\set{a, v_1, v_2}$, and let $B$ be the triangle that includes vertex $b$, $v_3$, and some other vertex. First, we create $W_1 = B \cup r_G(B)$ and update the graph to $G' = G - W_1$. We then make the second vertex set $W_2 = A \cup r_{G'}(A)$, and again update the graph to $G' - W_2$. The vertex sets $W_1, W_2$ are added to $\mathcal{H}$.
	\end{subcase}

	\begin{subcase}[$a$ is only adjacent to $v_1$ and $b$ is only adjacent to $v_2$.]\label{case:adjToTwoDiffVtxs}
		 This subcase is visualized in Figure~\ref{fig:exactalgcase3}. We try to find the first triangle $T$ in the direction of either $a$ or $b$, that satisfies $|r(T)| \neq 2$. The direction does not matter, so let us assume that we search in the direction of $a$. We first check the triangle $A$ where $a$ and $v_1$ are part of. If $|r(A)| = 2$, then we continue to the next triangle. Let $T_s$ denote the $s$-th checked triangle, where $T_0 = T$, $T_1 = A$, and so on. Two cases can occur with $T_s$: either $|r(T_s)| \neq 2$, or $T_{s+1} = T_0$ and the sequence of triangles form a cycle.
		
		\begin{subsubcase}[$|r(T_s)| \neq 2$]\label{case:chainNoCircle}
			We start by creating $W_1 = T_s \cup r_G(T_s)$, update the graph to $G_1 = G - W_1$, and continue to $T_{s-2}$. This triangle will now have $|r_{G_1}(T_{s-2})| = 1$. So, we create the vertex set $W_2 = T_{s-2} \cup r_{G_1}(T_{s-2})$, and update the graph again, $G_2 = G_1 - W_2$. We keep on doing this with steps of two.
			
			If $s$ is odd, we eventually end up in $T_1$ with the reduced graph $G_{s/2 - 1}$. We create the final vertex set $W_{s/2} = T_1 \cup r_{G_{s/2}}(T_1)$, and update the graph for the last time $G_{s/2} = G_{s/2 - 1} - W_{s/2}$. The vertex sets $W_1, ..., W_{s/2}$ are added to $\mathcal{H}$.
			
			If $s$ is even, we end up in $T$ (or $T_0$) with graph $G_{s/2 - 1}$. The final vertex set will therefore be $W_{s/2} = T \cup r_{G_{s/2 - 1}}(T) = T \cup \set{b}$. Again, we update the graph for the final time $G_{s/2} = G_{s/2 - 1} - W_{s/2}$. 
		\end{subsubcase} 
		
		\begin{subsubcase}[$T_{s+1} = T_0$]\label{case:chainIsCircle}
			If $s$ is odd, we define $W_1$ as $T_s$ together with the vertex that is in $T_{s-1}$ and in $r_G(T_s)$. So more formal, $W_1 = T_s \cup \big(T_{s-1} \cap r_G(T_s) \big)$. We then update the graph to $G_1 = G - W_1$, and continue to $T_{s-2}$. The size of $r_{G_1}(T_{s-2})$ has now become one, and therefore we create the next vertex set $W_2 = T_{s-2} \cup r_{G_1}(T_{s-2})$. Again, we update the graph and continue in steps of two. Eventually, we end up in $T_1$ with as graph $G_{s/2 - 1}$. In $G_{s/2 - 1}$ only one vertex of $T$ is left. So, we create the final vertex set $W_{s/2} = T_1 \cup r_{G_{s/2-1}}(T_1)$, and update the graph for the last time, $G_{s/2} = G_{s/2-1} - W_{s/2}$. The vertex sets $W_1, ..., W_{s/2}$ are added to $\mathcal{H}$.
			
			If $s$ is even, we define $W_1 = T_s \cup r_G(T_s) \cup T$. We update the graph to $G_1 = G - W_1$ and continue to $T_{s-2}$. The size of $r_{G_1}(T_{s-2})$ has now become one, and therefore we create the next vertex set $W_2 = T_{s-2} \cup r_{G_1}(T_{s-2})$. Again, we update the graph and continue in steps of two. Eventually, we end up in $T_2$ with the reduced graph $G_{s/2 - 1}$. Here, we create the final vertex set $W_{s/2} = T_2 \cup r_{G_{s/2 - 1}}(T_2)$ and update the graph for the last time, $G_{s/2} = G_{s/2 - 1} - W_{s/2}$. The vertex sets $W_1, ..., W_{s/2}$ are added to $\mathcal{H}$.
		\end{subsubcase}
	\end{subcase}
	
	\begin{figure}
		\centering
		\includegraphics[width=0.7\linewidth]{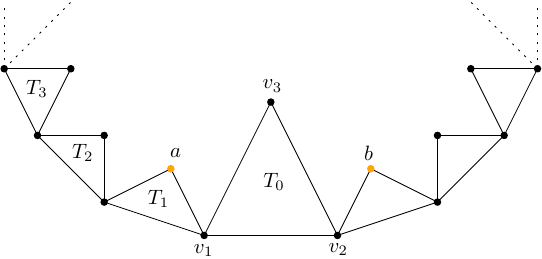}
		\caption{Case \ref{case:adjToTwoDiffVtxs} visualized. The vertices in orange correspond to the vertices in $r(T_0)$.}
		\label{fig:exactalgcase3}
	\end{figure}
\end{case}

\medskip
To prove that the above described procedure is correct, we first establish an auxiliary lemma that will be used multiple times in the subsequent lemma proving the correctness of the whole procedure. 

\begin{lemma}\label{lem:rOfTheRisNotAdjacent}
	Let $G = (V, E)$ be a graph where every vertex is part of some triangle. Consider an arbitrary triangle $T$ in $G$, let $x$ be a vertex in $r(T)$, and let $X$ denote the triangle that contains $x$ and at least one vertex from $T$. If a vertex $y \in r(X)$ is adjacent to $x$, then $y \in r(T)$.
\end{lemma} 

\begin{proof}
	We consider a vertex $y \in r(X)$ which is adjacent to $x$. Since every vertex in $G$ belongs to a triangle, $y$ must be part of some triangle, $Y$. Because $y \in r(X)$, at least one of the other two vertices of $Y$ must be a vertex of $X$. Let us consider the two possible scenarios for triangle $Y$:
	
	\setcounter{case}{0}
	\begin{case}[$x \notin Y$]
		In this case, $Y$ must be of the form $\set{y, u, v}$ where $u$ and $v$ are some vertex. Because $y \in r(X)$ at least $u$ or $v$ must be in $X \setminus \set{x}$.
		If $u \in X \setminus (\set{x} \cup T)$, then the triangle $\set{x, y, u}$ would also exist in $G$. Because $u \notin T$, $\set{x, y, u}$ would exist in $G-T$ as well. This contradicts that $x \in r(T)$.
		Therefore, the assumption that $u \in X \setminus (\set{x} \cup T)$ must be false. Consequently, if $u \in X \setminus \set{x}$, then $u$ must be in $X \cap T$. If $T$ is removed from the graph, then the triangle $Y$ would not exist anymore. Thus, $y \in r(T)$.
	\end{case}
	
	\begin{case}[$x \in Y$]
		In this case, $Y$ must be of the form $\set{x, y, u}$ for some vertex $u$. If $u \in V \setminus (T \cup \set{x, y})$, then the triangle $Y$ would also exist in $G - T$, which contradicts $x \in r(T)$. Therefore, $u$ must be in $T$, and if $T$ is removed from the graph, then $Y$ does not exist anymore. Hence, $y \in r(T)$.
	\end{case}
	
	\medskip In both cases, if $y \in r(X)$ and $y$ is adjacent to $x$, then $y \in r(T)$. This concludes the proof.
\end{proof}

Using this lemma, we prove that the whole procedure above adds the correct vertex sets.

\begin{lemma}\label{lem:procIsCorrect}
	The procedure above will only add fast vertex sets to $\mathcal{H}$. 
\end{lemma}

\begin{proof}
	To prove this, we investigate all the possible vertex sets that can be added to $\mathcal{H}$ by the above described procedure. We use the definitions as described in the procedure.
	
	\medskip\noindent\textbf{Case~\ref{case:exactAlgEasyCase}}: The condition for this case is that either two or more vertices in $r(T)$ are adjacent to the same vertex in $T$, or $|r(T)| = 1$. This ensures that the set $T \cup r(T)$, which is added to $\mathcal{H}$, is a fast vertex set.
	
	\medskip\noindent\textbf{Case~\ref{case:thrdVtxsDistinct}}: In this case, the vertex sets $W_i = T_i - r(T_i)$ are constructed, where $i \in \set{1,2,3}$. The vertices of $T_i$ are defined as $\set{v_{i, 1}, v_{i, 2}, v_{i, 3}}$, where $v_{i,1} \in T$, and $v_{i, 2} \in r(T)$. As mentioned in the procedure, this case is visualized on the left in Figure~\ref{fig:exactalgcase2}. We investigate $W_1$ by first considering the vertices of $r(T_1)$. These vertices cannot be adjacent to only $v_{1,1}$ in $T_1$. If they were, they would have been included in $r(T)$, which contradicts the case definition.
	By Lemma~\ref{lem:rOfTheRisNotAdjacent}, we know that for a vertex $y \in r(T_1)$ to be adjacent to $v_{1,2}$, it must also be in $r(T)$. But the other two vertices in $r(T)$ are not in $r(T_1)$. This means that any vertex in $r(T_1$) is not adjacent to $v_{1,2}$.
	Consequently, we know that $T_1$ has two vertices ($v_{1,1}$ and $v_{1,2}$) that are not adjacent with any vertex in $r(T_1)$. Therefore, $W_1 = T_1 \cup r(T_1)$ is a fast vertex set. The same reasoning applies for $W_2$ and $W_3$.
	
	\medskip\noindent\textbf{Case~\ref{case:commonThirdVtx}}: The case defines the common vertex as $u$ and the triangle, to which $u$ belongs, as $P$. The vertices of $r_G(T)$ that are adjacent to $u$, are also in $r_G(P)$ if Case~\ref{case:thrdVtxsDistinct} is not applicable. Therefore, the vertex set $W_1 = P \cup r_G(P)$ is a fast vertex set because at least two vertices of $r_G(P)$ are adjacent to the same vertex in $P$. By deleting $W_1$ from $G$, at least two vertices of $T$ are not adjacent to anything in $r_{G'}(T)$. Thus, $W_2$ is also a fast vertex set.
	
	\medskip\noindent\textbf{Case~\ref{case:vtxConnectToTwoOfT}}: The case defines $\set{v_1, v_2, v_3}$ as the vertices of $T$, and $a$ and $b$ as the vertices in $r(T)$. $A$ denotes the triangle $\set{a, v_1, v_2}$ and $B$ denotes the triangle $\set{b, v_3, u}$, where $u$ is some other vertex in $G$. We know that no vertex from $r(B)$ is only adjacent to $v_1$ of $T$. Otherwise, these vertices should have been in $r(T)$ as well, which contradicts the case definition. By Lemma~\ref{lem:rOfTheRisNotAdjacent}, we know that for a vertex $y \in r(B)$ to be adjacent to $b$, it must be in $r(T)$. But the other vertex in $r(T)$ is not in $r(B)$, so any vertex in $r(B)$ is not adjacent to $b$. Therefore, the vertices of $r(B)$ (if there are any) are only adjacent to $u$, making $W_1$ a fast vertex set. 
	After deletion of $W_1$ from $G$, we know that at least two corners ($v_1$ and $v_2$) of $A$ are not adjacent to any vertices of $r(A)$. This makes $W_2 = A \cup r(A)$ a fast vertex set. 
	
	\medskip\noindent\textbf{Case~\ref{case:chainNoCircle}}: The first vertex set created in this case is defined as $W_1 = T_s \cup r_G(T_s)$. 
	
	Consider the triangle $T_l$ that is $l$ steps away from $T_0$. Assume that $T_l \cup r(T_l)$ is a slow vertex set where $|r(T_l)| = 3$ and every vertex of $T_l$ is adjacent to one unique vertex of $r(T_l)$. Observe that the procedure will never add $T_l \cup r(T_l)$ to $\mathcal{H}$, as it will never reach triangle $T_l$. This is because the vertices of $T_l$ will not be in $r(T_{l-1})$. So, if the procedure reaches $T_{l-1}$ then $|r(T_{l-1})| = 1$, which implies that $T_s = T_{l-1}$, and $T_l$ is never reached. Therefore, in this case, the first vertex set added to $\mathcal{H}$ will not be the slow vertex set $W$ where every vertex of the triangle $T \subseteq W$ is adjacent to one unique vertex in $W \setminus T$ and $|W \setminus T| = 3$.  
	
	Also, the first vertex set, $W_1$, cannot be the second slowest vertex. This is because the procedure finds the first triangle $T_s$ that has $|r(T_s)| \neq 2$. As a result, $W_1$ can never correspond to the vertex set where two vertices of the triangle $T \subseteq W_1$ are adjacent to exactly one unique vertex in $W \setminus T$ and $|W \setminus T| = 2$. 
	
	All the following vertex sets $W_i$, where $i > 1$, are fast because the deletion of $W_{i-1}$ from $G_{i-2}$ ensured that the corresponding triangle $T_j$ has $|r_{G_{i-1}}(T_j)| = 1$. 
	
	If $s$ is \textit{odd}, we eventually end up in $T_1$ with the reduced graph $G_{s/2 - 1}$. Here, the only vertex in $r_{G_{s/2}}(T_1)$, will be a vertex of $T$ (or $T_0$). So, the final vertex set $W_{s/2}$ is fast. If $s$ is \textit{even}, we end up in $T$ (or $T_0$) with graph $G_{s/2 - 1}$. The only vertex in $r_{G_{s/2 - 1}}(T)$ is $b$ so $W_{s/2}$ is fast as well.
	
	\medskip\noindent\textbf{Case~\ref{case:chainIsCircle}}: In the case that $s$ is \textit{odd}, then $W_1$ is a fast vertex set because $|T_{s-1} \cap r_G(T_s)| = 1$. As mentioned before, all the following vertex sets $W_i$, where $i > 1$, are fast vertex sets. Eventually, we end up in $T_1$ where the only vertex in $r_{G_{s/2 - 1}}(T_1)$ is a vertex of $T$ (or $T_0$). Therefore, $W_{s/2}$ is also a fast vertex set. 
	
	If $s$ is \textit{even}, then $W_1$ contains $T$. Since two vertices of $T$ are adjacent to the same vertex of $T_s$, it follows that $W_1$ is a fast vertex set. All the following vertex sets $W_i$, where $i > 1$, are fast vertex sets as mentioned before. Eventually, we end up in $T_2$ where the only vertex in $r_{G_{s/2 - 1}}(T_2)$ is a vertex of $T_1$. So, $W_{s/2}$ is a fast vertex set as well.
\end{proof}

\begin{lemma}\label{lem:invariantHolds}
	The procedure above will maintain the invariant.
\end{lemma}

\begin{proof}
	Case~\ref{case:exactAlgEasyCase}, Case~\ref{case:exactAlgThreePerCorner}, Case~\ref{case:vtxConnectToTwoOfT} and Case~\ref{case:chainNoCircle} of the procedure described above, all create the vertex sets $W_i$ by the union of some vertex set $U$ together with $r(U)$. By the definition of $r$, we know that $G - (U \cup r(U))$ is a graph where all the vertices are part of some triangle. So, the lemma holds for these cases.
	
	For Case~\ref{case:chainIsCircle} it is a bit different. If $s$ is odd, then $W_1 = T_s \cup \big(T_{s-1} \cap r_G(T_s) \big)$. After deleting $W_1$ from the graph $G$, there will be a vertex $v$ in the graph that is not part of a triangle anymore. This vertex $v$ was before deletion part of the triangle $T$. As the procedure describes, this vertex $v$ is added in the final vertex set $W_{s/2}$, because $v$ is also part of $r(T_1)$. 
	
	If $s$ is even, the same behaviour occurs. Because all the vertices of $T$ are in $W_1$, there is a vertex of $T_1$ that after deletion, is not part of a triangle (In Figure \ref{fig:exactalgcase3} denoted as vertex $a$). But this vertex is put in de final vertex set $W_{s/2}$, because this vertex is also part of $r(T_2)$.
\end{proof}

\begin{lemma}
	Let $G$ be a graph where every vertex is part of a triangle. We can decompose the graph into disjoint vertex set family $\mathcal{H} = \set{W_1, ..., W_s}$, where no $W_i$ is a slow vertex set.
\end{lemma}

\begin{proof}
	By Lemma \ref{lem:invariantHolds}, we know that the above procedure can be applied repeatedly until the graph $G$ becomes empty. Lemma \ref{lem:procIsCorrect} proves that all vertex sets added to $\mathcal{H}$ by the above described procedure are fast vertex sets. Because $\mathcal{H}$ is initially empty, we can conclude that every vertex set in $\mathcal{H}$ is a fast vertex set.
\end{proof}

With the algorithm above we have shown how we can divide the graph into vertex sets that are all `fast'. Clearly, the above procedure is polynomial, and therefore the time complexity simply shifts to the next slowest vertex set. This is the vertex set $W$ where every vertex of the triangle $T \subseteq W$ is adjacent to two unique vertices in $W - T$. So, the vertex cost of $W$ is $(3 \cdot 1.3645^4)^{1/9} \approx 1.2972$. This all together concludes,

\begin{theorem}
	There is a $\OStar{1.2972^n}$ algorithm for \scs.
\end{theorem}
	\section{Minimum degree $c \cdot n$, where $c < 1$}\label{sec:minDegCN}
We consider the special case of \scs where the input graph is required to have a minimum degree $\delta = c \cdot n$, for a constant $c < 1$.
First, we establish an upper bound on the minimum degree of a graph that guarantees it does not contain a stable cutset.
This is followed by a polynomial-time algorithm 
that solves 
\scs on graphs with minimum degree $\delta \geq \tfrac{1}{2}n$. Finally, a similar kernelisation algorithm is presented for $\delta = \tfrac{1}{2}n - k$.

\subsection{Upper bound for minimum degree}
Chen and Yu \cite{ANoteOnFragileGraphs} gave an upper bound for the number of edges for which it is certain that a stable cutset exists in the graph. Here, we will do the opposite with a different parameter, investigating the minimum degree for which a graph is guaranteed \emph{not} to contain a stable cutset. To achieve this, we construct a graph that both contains a stable cutset and has the largest possible minimum degree.

\begin{theorem}\label{thrm:maxMinDegSCS}
	A graph $G$ with minimum degree $\delta > \tfrac{2}{3}(n - 1)$ does not contain a stable cutset.
\end{theorem}

\begin{proof}
	Consider the graph $G$ that contains a stable cutset $S$. Let $A$ and $B$ denote the vertex sets that are separated. We want to maximize the degree of every vertex in the graph. Vertices in $A$ can only be adjacent to other vertices in $A$ and vertices in $S$. Therefore, we know that the degree of any vertex $v$ in $A$ is tightly upper bounded by $deg(v) \leq |A| - 1 + |S|$. This same reasoning can be applied for any vertex $u$ in $B$: $deg(u) \leq |B| - 1 + |S|$. Because $S$ is stable, the degree of any vertex $w \in S$ is upper bounded by $deg(w) \leq |A|+|B|$. A graph where the degree of every vertex is maximized can be seen in Figure \ref{fig:maxmindegscs}.
	
	If any vertex does not comply with the degree constraint of his vertex set, then the graph does not contain a stable cutset. The minimum degree of a graph forces all the degrees of the vertices to be a specific value. Therefore, if
	\begin{align*} 
		\delta &>  |A| + |B| \\ 
		\delta &>  |A| + |S| - 1 \\
		\delta &>  |B| + |S| - 1 
	\end{align*}
	then we know that $G$ does not contain a stable cutset. The three vertex sets cover all vertices of $G$ and therefore $|A| + |B| + |S| = n$. If we combine this fact with the second and third statement above, then we get $|B| > n - \delta - 1$ and $|A| > n - \delta - 1$, respectively. If we fill this in the first statement above we retrieve $\delta > \tfrac{2}{3}(n - 1)$. Therefore, if the minimum degree in $G$ is larger than $\tfrac{2}{3}(n - 1)$, then we know that $G$ does not contain a stable cutset.
\end{proof}

\begin{corollary}
	A graph that contains a \scs has minimum degree $\delta \leq \tfrac{2}{3}(n - 1)$
\end{corollary}

\begin{figure}
	\centering
	\includegraphics[width=0.65\linewidth]{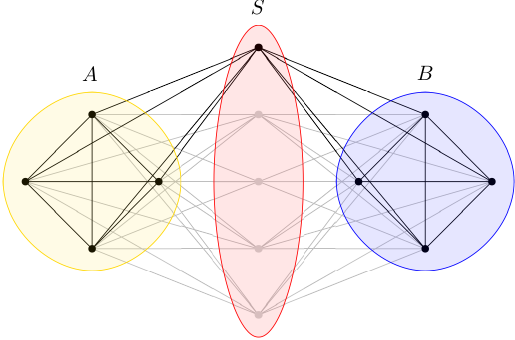}
	\caption{Example of a graph that contains a stable cutset and where the minimum degree is maximized. Not all the edges are coloured black for the clarity of the figure.}
	\label{fig:maxmindegscs}
\end{figure}

\subsection{Polynomial-time algorithm for $\delta \geq \tfrac{1}{2}n$}\label{sec:polyAlgForMinDegHalfN}
Consider the graph $G = (V,E, p)$ as described in Section \ref{sec:scsAsAnnot} with a minimum degree of $\delta \geq \tfrac{1}{2}n$. We present in this section an algorithm that in polynomial time decides whether $G$ contains a stable cutset.

The algorithm considers all possible non-adjacent vertex pairs in $G$, which can be enumerated in $O(n^2)$ time. Suppose that the algorithm is currently investigating non-adjacent vertex pair $(u, v) \in V \times V$. We know by Lemma \ref{lem:commonNeighMinDegHalfN} that $N(u) \cap N(v) \neq \emptyset$.

\begin{lemma}\label{lem:commonNeighMinDegHalfN}
	Consider a graph $G = (V, E)$ with minimum degree $\delta \geq \tfrac{1}{2}n$. Every vertex pair $(u, v) \in V \times V$ that is non-adjacent has at least 1 common neighbour.
\end{lemma}

\begin{proof}
    Let $G$ and the vertex pair $(u,v) \in V\times V$ be as in the lemma above.
	%Consider having the graph $G$ and the vertex pair $(u, v) \in V \times V$ as described in the lemma.
    Both $u$ and $v$ are adjacent to at least $\tfrac{1}{2}n$ other vertices if $n$ is even and at least $\tfrac{1}{2}(n-1)$ other vertices
    if $n$ is odd. There are $n-2$ vertices left which are possible neighbours for $u$ and $v$. Thus, by the pigeon hole principle, there are at least 2 vertices in $V \setminus \set{u, v}$ that are adjacent to $u$ and $v$ if $n$ is even. And by the same principle, there is at least one common neighbour for $u$ and $v$ if $n$ is odd.
\end{proof}

We label $p(u) = \set{A}$ and $p(v) = \set{B}$. As a result, by Lemma~\ref{lem:annGraphRules}, all the vertices in $N(u) \cap N(v)$ can be labelled with $\set{S}$. We define the numerical variable $r = |N(u) \cap N(v)|$. Also, we define the vertex set $X = N(u) \setminus \big(N(u) \cap N(v)\big)$, and $Y = N(v) \setminus \big(N(u) \cap N(v)\big)$. Finally, we have the remaining vertices set $F = V - X - Y - (N(u) \cap N(v))$. Figure \ref{fig:polyalgmindeghalfn} gives a visual presentation of the vertex sets.

\begin{claim}\label{clm:polyAlgSizesSets}
	Given the above definitions: $|X| \geq \tfrac{1}{2}n - r$, $|Y| \geq \tfrac{1}{2}n - r$ and $|F| \leq r - 2$	
\end{claim}

\begin{proof}
	The vertex set $X$ is defined as the neighbours of $u$ minus the common neighbours of $u$ and $v$. We know by the minimum degree precondition that $u$ has at least $\tfrac{1}{2}n$ neighbours. We defined $r$ as the cardinality of the common neighbours, and therefore we know that $|X| \geq \tfrac{1}{2}n - r$. The same reasoning can be used for $Y$. 
	
	$F$ is defined as the remaining vertices. To calculate the cardinality of $F$, we take all the vertices and subtract the sizes of the other vertex sets. So, $|F| = n - 2 - r - |X| - |Y| \leq n - 2 - r - (\tfrac{1}{2}n - r) - (\tfrac{1}{2}n - r) = r - 2$.
\end{proof}

The algorithm now applies the rules of Lemma~\ref{lem:annGraphRules}
until it is no longer possible. If at some point there exists a vertex $v$ with $p(v) = \emptyset$, then the algorithm stops and continues to the next pair of vertices $(u', v')$. If the algorithm does not stop and the rules of Lemma~\ref{lem:annGraphRules} cannot be applied anymore, then the following claims hold: 

\begin{figure}
	\centering
	\includegraphics[width=1\linewidth]{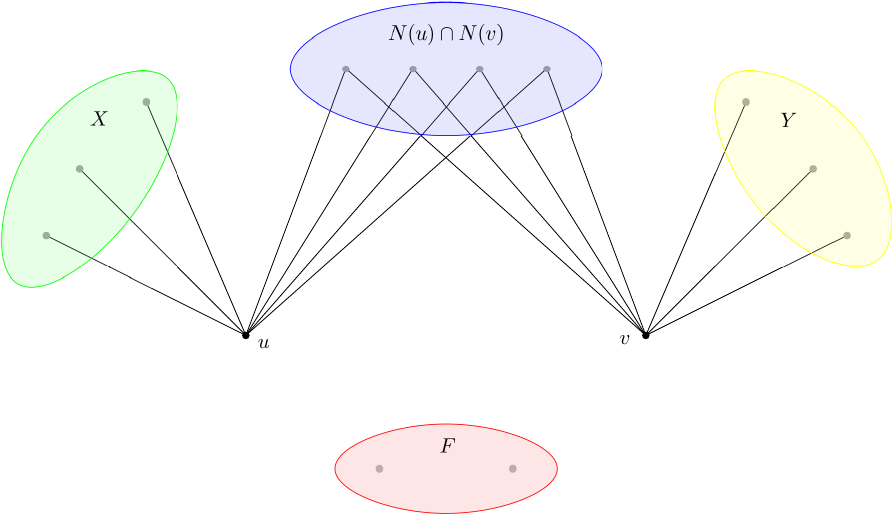}
	\caption{The division of the graph into vertex sets after choosing a non-adjacent pair $(u,v)$. Note, only the edges of $u$ and $v$ are drawn.}
	\label{fig:polyalgmindeghalfn}
\end{figure}

\begin{claim}\label{clm:polyAlgHowManyLabeled}
	At least $\tfrac{1}{2}n + 2$ vertices are annotated with one possible label.
\end{claim}

\begin{proof}
	Consider the vertex $s \in N(u) \cap N(v)$. We know that $s$ is adjacent to $u$ and $v$, but also to at least $\tfrac{1}{2}n - 2$ other vertices. Because $|F| \leq r - 2$ (Claim \ref{clm:polyAlgSizesSets}), at least $\tfrac{1}{2}n - 2 - (r - 2) = \tfrac{1}{2}n - r$ vertices of $X$ or $Y$ are adjacent to $s$.
	
	Now consider the vertex $x \in X$ that is adjacent to $s$. By definition of $X$, we know that $x$ is also adjacent to $u$. This means that $x$ is adjacent to a vertex with label $\set{A}$ and a vertex with label $\set{S}$. By Lemma \ref{lem:annGraphRules}, $p(x) = \set{A}$. For a vertex $y \in Y$ adjacent to $s$, the same reasoning applies: $p(y) = \set{B}$. Therefore, at least $\tfrac{1}{2}n - r$ vertices are labelled with $A$ or $B$. 
	
	Combining this with the $r$ vertices labelled with $\set{S}$, and the fact that $u$ and $v$ are labelled with $\set{A}$ and $\set{B}$, respectively, we obtain a total of at least $r + 2 + \tfrac{1}{2}n - r = \tfrac{1}{2}n + 2$ vertices that are labelled with one label.
\end{proof}

\begin{claim}\label{clm:polyAlgAllVtxsTwoPossibleLabels}
	Each vertex is labelled with at most two possible labels.
\end{claim}

\begin{proof}
	We know from Claim \ref{clm:polyAlgHowManyLabeled} that at least $\tfrac{1}{2}n + 2$ vertices are labelled with one label. 
	Therefore, at most $\tfrac{1}{2}n - 2$ vertices can still potentially be annotated with two or three possible labels. 
	By the claims and description above, we know that all the vertices \emph{not} in $F$ have at most two possible labels. The lemma already holds for these vertices. Let us consider a vertex $f \in F$. This vertex is adjacent to at least $\tfrac{1}{2}n$ other vertices because of the minimum degree property of the graph. Besides $f$ there are at most $\tfrac{1}{2}n - 3$ other vertices that are not labelled with one label. It follows that $f$ is adjacent to at least $\tfrac{1}{2}n - (\tfrac{1}{2}n - 3) = 3$ vertices that are labelled with one label. 
	
	This reasoning can be applied for all the vertices in $F$. All the vertices in $X$ and $Y$ are adjacent to $u$ or $v$, making them all adjacent to a vertex labelled with one label. Therefore, by Lemma \ref{lem:annGraphRules}, all the vertices in $G$ have at most two possible labels.
\end{proof}

With these claims we can prove the following.

\begin{theorem}\label{thrm:halfDegreePolyAlgo}
	There exists a polynomial-time algorithm for \scs on graphs with minimum degree $\delta \geq \tfrac{1}{2}n$.
\end{theorem}

\begin{proof}
	Consider the graph $G$ with minimum degree $\delta \geq \tfrac{1}{2}n$. With the algorithm described above we know by Claim \ref{clm:polyAlgAllVtxsTwoPossibleLabels} that any vertex $v$ in $G$ has $|p(v)| \leq 2$. We now convert $G$ to a (3,2)-CSP instance as described in Section \ref{sec:annTo32CSP}. Lemma~\ref{lem:decreaseValuesBE} enables us to reduce the (3,2)-CSP instance, allowing it to be solved in polynomial time by the solver of Beigel and Eppstein \cite{3ColoringInTime1.3289^n}. 
\end{proof}

\subsection{Kernelisation for $\delta = \tfrac{1}{2}n - k$}
In the previous section, we presented a polynomial-time algorithm for graphs with minimum degree $\delta \geq \tfrac{1}{2}n$. Now, we investigate the graphs where the minimum degree is slightly lower, namely $\delta = \tfrac{1}{2}n - k$, where $k \in \N$ is a parameter. In this setting, we develop a kernelisation algorithm that reduces the instance size to a polynomial in $k$, while preserving the existence of a stable cutset.
The algorithm is very similar to the algorithm presented in Section \ref{sec:polyAlgForMinDegHalfN}. Therefore, we will skip certain duplicate explanations to enhance readability. 

Because Lemma \ref{lem:commonNeighMinDegHalfN} is not applicable in this context, we instead consider all pairs of non-adjacent vertices $(u,v) \in V \times V$ such that the distance between $u$ and $v$ is exactly two. This distance constraint forces $N(u) \cap N(v)$ not to be empty. Again, we label $p(u) = \set{A}$, $p(v) = \set{B}$, and all the vertices in $N(u) \cap N(v)$ with $\set{S}$. Also, we use the same definitions for the vertex sets.

\begin{claim}\label{clm:kernSetSizes}
	Given the above definitions: $|X| \geq \tfrac{1}{2}n - k - r$, $|Y| \geq \tfrac{1}{2}n - k - r$ and $|F| \leq 2k + r - 2$.
\end{claim}

\begin{proof}
	The proof is the same as that of Claim~\ref{clm:polyAlgSizesSets}, with minimum degree set to $\delta = \tfrac{1}{2}n - k$
\end{proof}

The algorithm now applies the rules explained in Section \ref{sec:scsAsAnnot}
until it is no longer possible. If at some point there exists a vertex $v$ with $p(v) = \emptyset$, then the algorithm stops and continues to the next pair of vertices $(u', v')$. 

Consider the vertex $s \in N(u) \cap N(v)$. We define the numerical variable $t$ which indicates how many vertices of $F$ are \emph{not} adjacent to $s$. 

\begin{claim}\label{clm:kernAnnotatedWithOne}
	At least $\tfrac{1}{2}n - 3k + t + 2$ vertices are annotated with one possible label.
\end{claim}

\begin{proof}
	The vertex $s$ defined above is besides $u$ and $v$ adjacent to at least $\tfrac{1}{2}n - k - 2$ other vertices. Therefore, the minimal number of neighbours of $s$ that belong to $X \cup Y$ is given by $|N(s)| - |N(s) \cap F|$. By Claim \ref{clm:kernSetSizes} and the definition of $t$ above, this translates to $(\tfrac{1}{2}n - k - 2) - (2k + r - 2 - t)$. Thus, at least $\tfrac{1}{2}n - 3k - r + t$ vertices of $X$ or $Y$ are adjacent to $s$. 
	
	By Lemma \ref{lem:annGraphRules}, we can now annotate all these vertices with $\set{A}$ or $\set{B}$, if the corresponding vertex is in $X$ or $Y$ respectively. 
	Thus, in total at least $r + 2 + \tfrac{1}{2}n - 3k - r + t = \tfrac{1}{2}n - 3k + t + 2$ vertices are labelled with one possible label.
\end{proof}

With these claims we prove, 

\begin{theorem}
	Let $G$ be a graph with minimum degree $\delta = \tfrac{1}{2}n - k$. There exists a polynomial-time algorithm that reduces the \scs instance $(G, k)$ to a (3,2)-CSP instance with a size polynomial in $k$, where the result is preserved.  
\end{theorem}

\begin{proof}
	By Claim~\ref{clm:kernAnnotatedWithOne}, we know that at least $\tfrac{1}{2}n - 3k + t + 2$ vertices are annotated with one possible label. Therefore, at most $\tfrac{1}{2}n + 3k - t - 2$ vertices have two or three possible labels annotated. Again, we know that all the vertices that are not in $F$ have at most two possible labels. So, we consider a vertex $f \in F$. This vertex is adjacent to at least $\tfrac{1}{2}n - k$ other vertices. Besides $f$ there are $\tfrac{1}{2}n + 3k - t - 3$ other vertices with two or three possible labels annotated.
	
	If $\tfrac{1}{2}n - k - (\tfrac{1}{2}n + 3k - t - 3) > 0$, which rewrites to $t > 4k - 3$, then we know that $f$ is adjacent to at least one vertex that is annotated with one possible label. This reasoning can be applied for all the vertices in $F$. By Lemma \ref{lem:annGraphRules}, we know that all the vertices in this case have at most two possible labels annotated.
	
	But, if $t \leq 4k - 3$, then it is not certain that $f$ is adjacent to a vertex with only one possible label. By definition of $t$ given above, we know that the $|F| - t$ other vertices in $F$ are adjacent to $s$. Therefore, at most $t$ variables have still three possible labels annotated. If we convert our instance to a (3,2)-CSP instance as described in Section \ref{sec:annTo32CSP}, and apply Lemma \ref{lem:decreaseValuesBE}, then we get a (3,2)-CSP instance with $t$ variables, which is $\mathcal{O}(k)$. 
\end{proof}
	\section{Minimum degree $c$}\label{sec:minDegC}
Chen et al.~\cite{MatchingCutKernelizationSingleExponentialTimeFPTAndExactExponentialAlgorithm} showed, for the closely related \problemfont{Matching Cutset} problem, a fast exact algorithm on graphs with minimum degree $c > 3$. In this section, we present a fast exact algorithm for \scs on graphs with minimum degree $\delta \geq 3$, which becomes fast when $\delta > 10$. First, we will prove that \scs remains \NPC when restricted to graphs with minimum degree $c$, where $c > 1$. Then we describe an exact algorithm that runs in $\OStar{\lambda^n}$ time, where $\lambda$ is the positive root of the polynomial $x^{\delta + 2} - x^{\delta + 1} - 6$.

\subsection{\NPC}

\begin{theorem}
	\scs remains \NPC when restricted to graphs with minimum degree $c$, where $c$ is a constant greater than 1.
\end{theorem}

\begin{proof}
	Given a constant $c > 1$, we reduce \scs, to \scs restricted to graphs with minimum degree $c$, in polynomial time. Consider a graph $G = (V, E)$. We create the graph $G'$ by attaching the vertices of every edge $e$ in $G$ with a $c-1$ sized clique, $Q_e$. So, given an edge $(u, v) \in E$, $Q_{(u,v)} \cup \set{u, v}$ is an induced clique of size $c+1$ in $G'$. The graph $G'$ has $m(c - 1)+ n$ vertices, $m \cdot \frac{c(c+1)}{2}$ edges, and a minimum degree of $c$. An example of this construction is visible in Figure \ref{fig:npcompletenessconstruction}.
	
	If you have some minimum stable cutset $S$ in $G$ and some edge $(u, v)$ in G, then at most one of the two vertices can be in $S$. We assume that $u \in S$. The clique $Q_{(u, v)}$ adjacent to $(u, v)$ in $G'$ is not adjacent to any other vertices in $G$. Therefore, the clique can be added to the connected component where $v$ is part of, without making any vertices of $A$ adjacent to vertices of $B$. So, if $G$ has a stable cutset, then $G'$ has a stable cutset. 
	
	Now, we consider having a minimum stable cutset $S$ separating $G'$ in $A$ and $B$. Also, we have some edge $(u, v)$ in $G'$, which is an edge in $G$. By construction of $G'$ we know that a clique $Q_{(u,v)}$ of size $c-1$ is adjacent to both $u$ and $v$. Let $x$ be a vertex part of $Q_{(u,v)}$. If $x \in S$, then $S \setminus \set{x}$ is a stable cutset in $G' - x$, still separating $A$ and $B$. This is because $S$ separates $G'$ into the connected components $A$ and $B$. So, deleting $x$ from $G'$ and $S$ will not make vertices of $A$ adjacent with vertices of $B$. 
	
	Consider that $x \in A$. In a clique only one vertex can be part of $S$. So, we know that $u$ or $v$ is in $A$, because $Q_{(u,v)} \cup \set{u, v}$ is an induced clique in $G'$. Therefore, deleting $x$ will not make the connected component $A$ empty, and thus we can say that $G' - x - S$ will still separate the graph into two connected components $A - x$ and $B$. The same reasoning can be applied for the case that $x \in B$. From this, we can conclude that if $G'$ has a stable set, then $G'$ without all $Q_e$ has a stable set. Since this reduced graph is exactly $G$, it follows that if $G'$ has a stable set, then so does $G$. 
	
	By proving both directions, we have showed $\mathcal{NP}$-hardness. The verification of a solution remains the same as for the normal \scs. Therefore, \scs when restricted to graphs with minimum degree $c$ is \NPC.
\end{proof}

\begin{figure}
	\centering
	\includegraphics[width=0.5\linewidth]{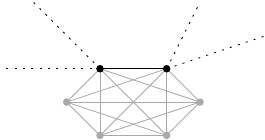}
	\caption{Visualization of how the instance for \scs with minimum degree $c > 1$ is constructed. In this example $c = 5$. The vertices and edges in black are of the \scs instance $G$, and the vertices in gray are added to create the instance $G'$.}
	\label{fig:npcompletenessconstruction}
\end{figure}

\subsection{Exact algorithm with Measure and Conquer}\label{sec:exactAlgMeasureConquer}
Now, we give a simple algorithm that solves the \scs problem for graphs with minimum degree $\delta \geq 3$. This is because a graph with minimum degree one is trivial, and graphs with minimum degree two can be reduced to graphs with $\delta \geq 3$. After that, we analyse the algorithm using \emph{Measure and Conquer} \cite{MeasureAndConquer} to get to a time complexity of $\OStar{\lambda^n}$, where $\lambda$ is the positive root of $x^{\delta + 2} - x^{\delta + 1} - 6$. This algorithm becomes relevant when $\delta \geq 11$ which gives a time complexity of $\OStar{1.2880^n}$ and therefore outperforms our own exact algorithm of Section \ref{sec:generalExactAlg}.
%Table~\ref{table:timeComp3Col} shows the value of $\lambda$ for some degrees.

The branching algorithm operates on an annotated graph $G = (V, E, p)$, as defined in Section~\ref{sec:scsAsAnnot}. First, it applies the rules from Lemma~\ref{lem:annGraphRules}. This will not do anything in the initial call, but in the subsequent recursive calls this will reduce the amount of possible labels per vertex significantly. Next, the algorithm checks for a vertex $v$ with $|p(v)| = 3$. If such a vertex exists then the algorithm branches into three subcases, assigning each possible label to $v$ in a separate branch. Each branch proceeds with a recursive call. If a vertex $v$ with $|p(v)| = 0$ is encountered, the algorithm terminates along that branch, as $v$ can no longer be labelled. Once no vertex with $|p(v)| = 3$ remains, the algorithm converts the annotated graph into a (3,2)-CSP instance, as described in Section~\ref{sec:annTo32CSP}. The result of the (3,2)-CSP solver will then determine whether $G$ contains a stable cutset. The pseudocode for this procedure can be seen in Algorithm \ref{alg:pseudoCodeBranchAlg}.

\begin{algorithm}
	\begin{algorithmic}
		\Function{BranchForSCS}{$G = (V, E, p)$}
			\State $G \gets \Call{applyRules}{G}$		
			\If{$\exists v \in V$ with $|p(v)| = 3$}
				\State $G_1 \gets \Call{BranchForSCS}{G : p(v) = \set{A}}$
				\State $G_2 \gets \Call{BranchForSCS}{G : p(v) = \set{B}}$
				\State $G_3 \gets \Call{BranchForSCS}{G : p(v) = \set{S}}$
				\If{any $G_i$ is flagged that it contains a \scs}
					\State \Return $G_i$
				\EndIf
			\Else
				\State $I \gets$ transform $G$ to (3,2)-CSP instance.
				\State \Return $\Call{32CSPSolver}{I}$
			\EndIf
		\EndFunction
	\end{algorithmic}
	\caption{Pseudocode of the branching algorithm that solves \scs.}
	\label{alg:pseudoCodeBranchAlg}
\end{algorithm}

To analyse this simple branching algorithm, we are going to use \emph{Measure and Conquer} \cite{MeasureAndConquer}. Unlike \emph{Branch and Reduce} analysis, which typically uses a standard measure such as the number of vertices, this technique focuses on choosing a smart measure that captures the progress of the algorithm better. This often leads to a tighter and more precise analysis of the running time.

In a standard branch and reduce analysis of our algorithm, the measure would likely be the number of vertices that do not have one possible label. However, reducing the number of possible labels for a vertex from three to two also represents meaningful progress in the algorithm because the corresponding variable for such a vertex can be deleted when the annotated graph is converted to a (3,2)-CSP instance. Therefore, we are going to use the following measure: $\kappa = w_2 n_2 + w_3 n_3$, where $n_i$ is the amount of vertices with $i$ possible labels, and $w_i$ is the corresponding weight. 

Let us analyse the branching rule of our algorithm using the new defined measure. Consider the vertex $v$ that is found in the graph with $|p(v)| = 3$. By Lemma~\ref{lem:annGraphRules}, we know that all the neighbours $u \in N(v)$ must satisfy $|p(u)| \geq 2$, because otherwise it would contradict with the assumption that $|p(v)| = 3$.

Assume that we have the vertices $x, y \in N(v)$, with $|p(x)| = 3$ and $|p(y)| = 2$. After branching and applying the rules of Lemma~\ref{lem:annGraphRules}, the weight of $v$ is reduced from $w_3$ to 0. The weight of $x$ is reduced from $w_3$ to $w_2$, and the weight of $y$ is reduced from $w_2$ to $0$ in two of the three branches.

Now, suppose that $v$ is adjacent to $t$ vertices with three possible labels, and $deg(v) - t$ vertices with two possible labels. Because the algorithm does not delete any vertices, we know that $deg(v) \geq \delta$. This gives the following branching factor: 
\[\tau_{w_3, t}\big(w_3 + t(w_3 - w_2), \quad w_3 + t(w_3-w_2)+ (\delta - t)w_2, \quad w_3 + t(w_3-w_2) + (\delta - t)w_2 \big)\]

Note that in the above expression for the branching factor, the weight $w_3$ appears only with a positive coefficient. Therefore, increasing $w_3$ will maximize its contribution to the decrease in measure. So, we set $w_3 = 1$ and update the branching factor accordingly:

\[\tau_{t}\big(1 + t(1 - w_2), \quad 1 + t(1-w_2)+ (\delta - t)w_2, \quad 1 + t(1-w_2) + (\delta - t)w_2 \big)\]

\begin{lemma}\label{lem:worstCaseBranchingLemma}
	Depending on the choice of $w_2$, the worst-case branching factor is either with $t = 0$ or $t = \delta$. 
\end{lemma}

\begin{proof}
	Consider the vertex $v$ with $|p(v)| = 3$, that is selected by the algorithm for branching. Each vertex in $N(v)$ is annotated with either two or three possible labels. If a vertex $x \in N(v)$ has $|p(x)| = 3$ then branching on $v$ causes a decrease in the measure of $1-w_2$. The smaller $w_2$ is, the more these type of neighbours contribute to the decrease of the measure.
	
	A neighbouring vertex $y \in N(v)$ with $|p(y) = 2|$ contributes a decrease of $w_2$ to the measure. A larger $w_2$ value means a greater contribution to the measure's decrease from such neighbours. Vertices annotated with two possible labels, will only reduce to one possible label in two out of the three branches formed by the branching rule. This implies that $w_2$ must be greater than $0.5$ to make a vertex annotated with two possible labels more interesting for the decrease in measure than a vertex annotated with three possible labels. 
	
	Thus, given any vertex $u \in N(v)$, depending on the choice of $w_2$, it is worse for the decrease of the measure if $u$ is labelled with either three or two possible labels. Therefore, the worst-case scenario for the decrease of the measure is that all the vertices in $N(v)$ are labelled with either two or three labels, depending on $w_2$.  
\end{proof}

By Lemma~\ref{lem:worstCaseBranchingLemma}, we know that the time complexity of our algorithm is determined by the maximum of the following two branching factors:
\begin{align*} 
	\tau_1 & \big(1 + \delta(1 - w_2), 1 + \delta(1-w_2), 1 + \delta(1-w_2) \big) \quad & (\forall u \in N(v): |p(u)| = 3) \\ 
	\tau_2 & \big(1, 1 + \delta w_2,    1 + \delta w_2 \big) \quad & (\forall u \in N(v): |p(u)| = 2)
\end{align*}

The task now is to find the value of $w_2$ that minimizes $\max{(\tau_1, \tau_2)}$. Observe that $w_2$ appears with a negative coefficient in $\tau_1$ and with a positive coefficient in $\tau_2$. Therefore, increasing $w_2$ causes $\tau_1$ to increase and $\tau_2$ to decrease, and vice versa. As a result, the minimum of $\max{(\tau_1, \tau_2)}$ is achieved when a value for $w_2$ is used such that the two expressions are equal. 

\begin{theorem}
	\scs can be solved in time $\OStar{\lambda^n}$ for a graph with minimum degree $\delta \geq 3$, where $\lambda$ is the positive root of $x^{\delta + 2} - x^{\delta + 1} - 6$.
\end{theorem}

\begin{proof}
	We know that the branching factor of $\tau_1$ is given by the positive root of the equation:
	\[\frac{1}{x^{1 + \delta(1 - w_2)}} + \frac{1}{x^{1 + \delta(1 - w_2)}} + \frac{1}{x^{1 + \delta(1 - w_2)}} = 1\]
	This simplifies to:
	\[x^{1+\delta(1-w_2)} = 3\]
	By taking the logarithm and rewriting this equation further, we retrieve:
	\[w_2 = 1 + \frac{1}{\delta} \Big(1-\frac{\log{3}}{\log{x}}\Big)\]
	We apply the same kind of rewriting with the equation for $\tau_2$ to retrieve:
	\[w_2 = \frac{\log{\frac{-2}{1-x}}}{\log{(x)} \delta}\]
	Since we want $\tau_1$ to have the same value as $\tau_2$, and because this must occur when both expressions yield the same value for $w_2$, we equate the two expressions:
	\[1 + \frac{1}{\delta} \Big(1-\frac{\log{3}}{\log{x}}\Big) = \frac{\log{\frac{-2}{1-x}}}{\log{(x)} \delta}\]
	Multiplying both sides with $\log{(x)}\delta$ and simplifying, leads to the following equation:
	\[x^{\delta + 2} - x^{\delta + 1} = 6\]
	Therefore, the positive root of this equation gives the branching factor $x$ such that $x = \tau_1 = \tau_2$, using the same $w_2$ value for $\tau_1$ and $\tau_2$.
\end{proof}

	\section{\problemfont{3-Colouring} for graphs with minimum degree $c > 1$}\label{sec:3ColMinDegC}
The current fastest exact algorithm for \problemfont{3-Colouring} on general graphs, recently proposed by Meijer, has a running time of $\OStar{1.3217^n}$ \cite{3ColoringInTime1.3217^n}.
For graphs with minimum degree $\delta \geq 3$ an algorithm based on dominating sets is presented that runs in $\mathcal{O^*}\Big[ \Big(3^{\tfrac{1+\ln(\delta + 1)}{1 + \delta}} \Big)^n \Big]$ \cite{3ColoringWithMinDegC}.

An interesting observation about the algorithm and time analysis of Section~\ref{sec:exactAlgMeasureConquer}, is that it also applies to the \problemfont{3-Colouring} problem. This is because the algorithm branches are based on the three possible labels of a vertex, which aligns with the 3-colouring constraint where each vertex must be assigned one of the three colours. 

Moreover, once a vertex is assigned a color, its neighbours can no longer be assigned that same color. Thus, if we branch on a vertex $v$ that currently has three possible colours, then each neighbour $u \in N(v)$ that also has three possible colours, will be reduced to two. This behaviour is the same as by the stable cutset algorithm of Section~\ref{sec:exactAlgMeasureConquer}. Similarly, for neighbours $u \in N(v)$ with only two possible colours remaining, their amount of possible colours will be reduced in only two of the three branches. Again, the same behaviour as by the stable cutset algorithm. So, we conclude,

\begin{theorem}
	\problemfont{3-Colouring} can be solved in time $\OStar{\lambda^n}$ for a graph with minimum degree $\delta \geq 3$, where $\lambda$ is the positive root of $x^{\delta + 2} - x^{\delta + 1} - 6$.
\end{theorem}

%HB: stukje hieronder geschreven

We conjecture that for
all values of the minimum degree $\delta>0$,
the time complexity we proved for our algorithm is smaller than the proved
time complexity for the
dominating sets-based algorithm from \cite{3ColoringWithMinDegC}, but we were unable to 
prove this analytically.

We have computed these
complexities for a number of
values of $\delta$, including those given in
Table \ref{table:timeComp3Col} 
and all positive $\delta\leq 100$.
In each of these cases, $
\lambda<\mu$, i.e., in these cases, our algorithm outperforms the algorithm from \cite{3ColoringWithMinDegC}.

It is also interesting to note that when the minimum degree $\delta\geq 10$, then our algorithm has a smaller time complexity than the 
algorithm from Meijer~\cite{3ColoringInTime1.3217^n}.

%Table \ref{table:timeComp3Col} presents the exponential factors for both functions for a few interesting minimum degrees. For all degrees $0 < \delta \leq 100$, we confirmed that our algorithm consistently achieves a lower time complexity than the dominating sets-based algorithm from \cite{3ColoringWithMinDegC}.

\begin{table}[ht]
	\begin{tabular}{cccccccccc}
		\hline
		\text{$\delta$}  & 3      & 15      & 25     & 42     & 50     & 75     & 100    & 642    & 8703     \\
		\text{$\mu$}     & 1.9259 & 1.2957	 & 1.1971 & 1.1294 & 1.1121 & 1.0801 & 1.0630 & 1.0128 & 1.0013   \\
		\text{$\lambda$} & 1.7069 & 1.2271  & 1.1519 & 1.1000 & 1.0866 & 1.0620 & 1.0488 & 1.0100 & 1.0010   \\
		\hline
	\end{tabular}
	\caption{$\mu$ is defined as $3^{\tfrac{1+\ln(\delta + 1)}{1 + \delta}}$ \cite{3ColoringWithMinDegC}, and $\lambda$ is the positive root of $x^{\delta + 2} - x^{\delta + 1} = 6$ for a few different minimum degrees $\delta$. Note that $\lambda < \mu$.}
	\label{table:timeComp3Col}
\end{table}

	\section{Conclusion}\label{sec:conclusion}
In this paper, we presented a new exact algorithm for \scs. The algorithm begins by partitioning the graph into disjoint vertex sets. Branching is then applied to these sets, as each is guaranteed to contain either a triangle or closed neighbourhood configuration. By branching, the amount of possible labels for the vertices involved in the configuration are reduced to less than or equal to two. As a result, the variables corresponding to these vertices can be deleted when the instance is transformed to a (3,2)-CSP instance. The remaining vertices in the vertex sets, which are not part of any branching configuration, will be solved by the (3,2)-CSP solver. Therefore, any improvement to the \problemfont{(3,2)-Constraint Satisfaction Problem} will directly enhance the performance of our algorithm.

We further examined the \scs problem in graphs with minimum degree $\delta = c \cdot n$, where $c < 1$. First, we established an upper bound for $c$ beyond which a stable cutset does not exist. By constructing a graph with maximum degree that still admits a stable cutset, we showed that no stable cutset exists when $c > \tfrac{2}{3}$. We then developed a polynomial-time algorithm that solves \scs on graphs with $\delta \geq \tfrac{1}{2}n$. The algorithm annotates two non-adjacent vertices $(u, v) \in V \times V$ with $\set{A}$ and $\set{B}$, after which the remaining vertices can be annotated with two or less possible labels. 

For graphs with minimum degree $\delta = \tfrac{1}{2}n - k$, we proposed a kernelisation algorithm that is structurally similar to the polynomial-time algorithm for $\delta \geq \tfrac{1}{2}n$. The algorithm either solves the instance in polynomial time or produces a (3,2)-CSP instance with at most $4k-3$ variables.

The final variant we studied is \scs in graphs with minimum degree $c$, where $c > 1$. We provided a $\mathcal{NP}$-completeness proof for this problem and introduced an exact algorithm that is analysed with the Measure and Conquer technique. The algorithm branches on vertices that have three possible labels. By labelling the vertex that is branched on with one of the possible labels, the possible labels of the neighbouring vertices can be reduced. Once all the vertices have two or less possible labels annotated, the instance is converted into a (3,2)-CSP instance, which can then be solved in polynomial time. By using Measure and Conquer, a time complexity of $\OStar{\lambda^n}$ is achieved, where $\lambda$ is the positive root of $x^{\delta + 2} - x^{\delta + 1} - 6$. 

Finally, we demonstrated that this exact algorithm for \scs in graphs with minimum degree $c$, where $c \geq 3$, can also be applied to solve \problemfont{3-Colouring} in graphs with minimum degree $c$. As a result, the best known algorithm for \problemfont{3-Colouring} with minimum degree is improved to a running time of $\OStar{\lambda^n}$, matching the bound obtained for \scs  in graphs with minimum degree $c$.
	%\newpage
	%\printbibliography

\bibliographystyle{abbrvurl}
\bibliography{references}
\end{document}